\documentclass[titlepage, reqno, 12pt]{amsart}
\usepackage{etoolbox}

\renewenvironment{titlepage}{%
  \thispagestyle{empty}\setcounter{page}{0}
  \vspace*{\fill}
}{%
  \vspace{3\baselineskip}
  \vspace*{\fill}
  \newpage
}
\makeatletter
\patchcmd{\@maketitle}
{\if@titlepage \newpage \else}
{\if@titlepage
 \vspace{\baselineskip}
 \else}
{}{}
\makeatother

\usepackage{amssymb}
\usepackage{mathrsfs}
\usepackage{dsfont}
\usepackage{enumerate, xspace}
\usepackage{changebar}
\usepackage{hyperref}
\usepackage{pdfsync}
\usepackage{color}
\usepackage{bm}
\usepackage{amsmath}

\newtheorem{thm}{Theorem}[section]

\newtheorem{cor}[thm]{Corollary}

\newtheorem{rem}[thm]{Remark}

\begin{document}
 \begin{titlepage}
   {\centering
     {\Large Financial accumulation implies ever-increasing\par wealth inequality\par}
  \vspace{2\baselineskip}
  Yuri Biondi and Stefano Olla\par
  \vspace{2\baselineskip}
}
  \bigskip

  \noindent
  \textbf{Abstract:}
 {\rm  Wealth inequality is an important matter for economic theory and policy. Ongoing debates have been discussing recent rise in wealth inequality in connection with recent development of active financial markets around the world. Existing literature on wealth distribution connects the origins of wealth inequality with a variety of drivers. Our approach develops a minimalist modelling strategy that combines three featuring mechanisms: active financial markets; individual wealth accumulation; and compound interest structure. We provide mathematical proof that accumulated financial investment returns involve ever-increasing wealth concentration and inequality across individual investors through time. This cumulative effect through space and time depends on the financial accumulation process and holds also under efficient financial markets, which generate some fair investment game that individual investors do repeatedly play through time. }
 \bigskip

 \noindent
 \emph{keywords}: {inequality, economic process, compound return, simple return, minimal institution}\\
\emph{AMS}[2010] : 91B70, \emph{JEL Classification}: C46, D31, D63, E02, E21\\
   \vspace{1cm}
  
  Yuri Biondi\\ CNRS, IRISSO, UMR 7170,
Universit\'e Paris-Dauphine, PSL,\\
75775 Paris-Cedex 16, France \\
{yuri.biondi@gmail.com}\\
orcid.org/0000-0002-5545-5550\\
\bigskip

Stefano Olla\\ CEREMADE, UMR CNRS 7534,
Universit\'e Paris-Dauphine, PSL\\
75775 Paris-Cedex 16, France \\
{olla@ceremade.dauphine.fr}\\
orcid.org/0000-0003-0845-1861\\
\end{titlepage}

\title[Financial accumulation and wealth inequality]{Financial accumulation implies ever-increasing wealth inequality}
\author{Yuri Biondi and Stefano Olla}
\begin{abstract}
  Wealth inequality is an important matter for economic theory and policy. Ongoing debates have been discussing recent rise in wealth inequality in connection with recent development of active financial markets around the world. Existing literature on wealth distribution connects the origins of wealth inequality with a variety of drivers. Our approach develops a minimalist modelling strategy that combines three featuring mechanisms: active financial markets; individual wealth accumulation; and compound interest structure. We provide mathematical proof that accumulated financial investment returns involve ever-increasing wealth concentration and inequality across individual investors most of the time. This cumulative effect through space and time depends on the financial accumulation process and holds also under efficient financial markets, which generate some fair investment game that individual investors do repeatedly play through time. 
\end{abstract}

\address{
Yuri Biondi\\ CNRS, IRISSO, UMR 7170,
Universit\'e Paris-Dauphine, PSL,\\
75775 Paris-Cedex 16, France \\
{yuri.biondi@gmail.com}}

%\author{Stefano Olla}

%\institute{
\address{
Stefano Olla\\ CEREMADE, UMR CNRS 7534
Universit\'e Paris-Dauphine, PSL
75775 Paris-Cedex 16, France \\
{olla@ceremade.dauphine.fr}}

\keywords{inequality, economic process, compound return, simple return, minimal institution}

\thanks{\emph{Acknowledgements:} This working paper benefited from informal discussion with the DYNAMETS research group (Dynamic Systems Analysis for Economic Theory and Society) based at University Paris Dauphine PSL. A previous version was presented at Yale University Research Workshop on 24 April 2018.
  We thank Shyam Sunder, Enrico Scalas, Mauro Gallegati, Yonatan Berman and Tae-Seok Jang for their comments and suggestions.
}

\maketitle

%\title[Financial accumulation and wealth inequality]{Financial accumulation implies ever-increasing wealth inequality}

 %\titlerunning{Financial accumulation and wealth inequality}

% \date{\today. {\bf File: {\jobname}.tex.}}
% \begin{abstract}
% Wealth inequality is an important matter for economic theory and policy. Ongoing debates have been discussing recent rise in wealth inequality in connection with recent development of active financial markets around the world. Existing literature on wealth distribution connects the origins of wealth inequality with a variety of drivers. Our approach develops a minimalist modelling strategy that combines three featuring mechanisms: active financial markets; individual wealth accumulation; and compound interest structure. We provide mathematical proof that accumulated financial investment returns involve ever-increasing wealth concentration and inequality across individual investors through time. This cumulative effect through space and time depends on the financial accumulation process and holds also under efficient financial markets, which generate some fair investment game that individual investors do repeatedly play through time. 
%   \end{abstract}

\section{Introduction}
\label{sec:introduction}

Wealth inequality is an important matter for economic theory and policy.
Current debates have been discussing recent rise in wealth inequality in connection
with recent development of active financial markets around the world
(\cite{Beck et al. 2007, Piketty 2013,Krugman 2013,Stiglitz 2012,Solow 2014}).
While some argued for the role of financial development in reducing income inequality and benefiting the poorer, others criticize the increasing concentration of financial capital and related income as the main source of raising inequality. Moreover, issues of wealth distribution were raised by the 99\% movement in US in the aftermath of the Global Financial Crisis of 2007-08, claiming that the increased financialisation of economy and society in recent decades has involved an increased appropriation of wealth by the richest 1\%
of the population at the detriment of the remaining 99\%,
implying a more unequal and allegedly unfair wealth distribution.

Existing literature on wealth distribution connects the origins of wealth inequality with a variety of drivers
(\cite{Bertola et al. 2006}, \cite{Snowdon and Vane 2005}).
Already at the beginning of the XX century, pointing to the concentration of wealth in economy and society,
the economist and sociologist V. Pareto suggested the so-called Pareto wealth distribution as an empirical regularity
(\cite{Pareto 1897}), while the economic statistician C. Gini developed ingenious statistical techniques to summarise
wealth inequality through the so-called Gini Index (\cite{Gini 1912}).
In this context, a relevant literature stream draws upon \cite{Champernowne 1953} and \cite{Rutherford 1955}
to develop an elegant formal modelling that explains wealth concentration and the Pareto wealth distribution
under conditions of financial market efficiency and the stochastic distribution of financial returns across individuals
that invest in that market (\cite{Levy 2005,Levy and Levy 2003}).
This modelling strategy considers financial investment as a multiplicative process
closely related to a Kesten process (\cite{Kesten 1973,Redner 1990}).

Existing literature and the current debate on wealth inequality foreshadow some connection between wealth concentration and the financial investment process through active financial markets. Our approach consists in disentangling the minimal components of this process in order to formalise its working and infer general results on its impact on economy and society. In particular, we identify two minimal institutions that surround individual investment strategies: compound return structure, and active financial markets.

On the one hand, compound return structure stands at the core of financial investment process.
It implies that individual investors keep reinvesting financial proceeds together with the previous capital stock
through time and circumstances.
On the other hand, active financial markets assure that this financial investment process
is played as a fair game, involving some market order over return-seeking individual strategies.
Under efficient financial markets, individual investors extract their realised returns from
the same 
distribution, while individual results remain {independent through space and time},
preventing arbitrage opportunities
(\cite{Samuelson 1965,Samuelson 1973,Fama 1995,Biondi and Righi 2017}
providing further literature review).

Our mathematical modelling shall formalise this financial investment process in line with both minimal institutions, in order to study its relationship with relative wealth concentration across individuals in economy and society. In particular, we shall prove that, let alone without countering forces, financial accumulation implies ever-increasing relative wealth inequality through time. We define financial accumulation as the peculiar financial investment process in which financial proceeds are systematically reinvested together with the previous capital stock through time and circumstances.

The rest of the article is organised as follows. The next section summarises our modelling strategy.
Section \ref{sec:model-proof-under} and  \autoref{sec:model-proof} provide mathematical proof of
the relationship between wealth concentration and the financial investment process respectively under compound
and simple return structures. Section~\ref{sec:limits-implications} develops some implications
of and perspectives on our results. Section \ref{sec:concluding-remarks} concludes.

\section{Modelling Strategy}
\label{sec:modelling-strategy}

Existing models of economic process generally combine a multiplicative factor called capital with an additive factor called labour, in order to explain economic growth and study the aggregate distribution of wealth
(\cite{Nirei and Souma 2007} providing further references).
Our modelling strategy focuses only on the capital factor,
in view to disentangle featuring financial investment mechanisms and their impact on economy and society.

Generally speaking, financial investment process under efficient capital markets implies that all and every individual investment strategies generate independent and independently distributed returns extracted from the same distribution for all investors and at every point of time, altough real financial markets factually show some correlation through time and cross-sectionally. Our modelling strategy relaxes this reference framework while featuring cumulative effects through space and time. Our results do not require financial returns to be independent and independently distributed over time and across agents. Only some decorrelation through time is required. Although some previous studies noticed the non-stationarity of wealth distribution driven by financial investment (\cite{Kalecki 1945, Biondi and Righi 2019}), its cumulative effects and implications were generally neglected.

Cumulative effects relate to the financial accumulation process that is embedded in compound return structure. Accordingly, every investor starts investing an initial capital value through time by realising one given return at each point of time. Individual wealth investment is then repeatedly submitted to a series of realised returns that singles out each and every individual investment portfolio. Compound return structure further assures that financial investment proceeds sum-up with the past-period capital stock to be reinvested in the next period, generating individual financial accumulation through time.

Our modelling strategy points to this cumulative effect by featuring its impact on both individual financial accumulation through time, and the relative wealth concentration across individuals at some point of time. Our minimalist modelling strategy enables disentangling key features of financial accumulation process that characterises a large set of real world mechanisms such as saving deposits and mutual investment funds. In this context, it must be stressed that compound return is generally taken as benchmarking reference for assessing financial investment performance.

\section{Model and proof under compound interest structure }
\label{sec:model-proof-under}

Consider a population of $n$ individual investors $i = 1, \dots, n$, holding financial wealth available
for investment $W_{i}(t)$ at some {time $t$. We consider here discrete time dynamics, so $t$ is integer valued.}

We formalise the financial investment process according to the familiar structure of compound return
where the wealth $W_{i}(t+1)$ of investor $i$ at time $t+1$ depends on the past wealth stock $W_{i}(t)$
invested at some return $r_{i,t+1}$ as follows:
\begin{equation}
  W_{i}(t+1)= W_{i}(t)  (1 + r_{i,t+1})
  \label{eq:1}
\end{equation}
where $\{\mathbf r_t = (r_{1,t}, \dots,  r_{n,t})\}_{t\in \mathbb N}$
is an infinite sequence of identically distributed random vectors in $\mathbb R^n$.
We do not require that the components of these random vectors are independent.
Without loss of generality, we may impose $r_{i,t}\ge -1$.
In fact, we do only impose conditions that are required for mathematical proof.
For sake of clarity, if $r_{i,t} \ge 0$,
the investor $i$ maintains or increases its past wealth stock that was invested at time $t-1$,
while that investor loses it at least partly if $r_{i,t} < 0$.

This setting generates the following cumulative process through time from initial time $t=1$ to time $t=T$:
\begin{equation}
  W_{i}(T) = W_{i}(0)  \prod_{t=1}^T \{ 1+ r_{i,t}\}= W_{i}(0)\exp \left\{ \sum_{t=1}^T  Z_{i,t} \right\},
  \label{eq:2}
\end{equation}
where $Z_{i,t} = \log[1+r_{i,t}]$.
This logarithmic transformation defines, for any $i=1,\dots,n$, the individual series of realised financial returns
$\{ Z_{i,t} ; \  t=1,\dots,T \}$ as a sequence of identically distributed random variables,
and we require the following conditions:
\begin{enumerate}[(i)]
\item strictly positive finite second moments $0< \mathbb E Z_{i,t}^2 < +\infty$ \label{var1},
\item  $|\mathbb E \left(Z_{i,t} Z_{j,t}\right)|^2 < \mathbb E Z_{i,t}^2 \mathbb E Z_{j,t}^2$ for $i\neq j$ and any $t\ge 1$.
  \label{var2}
% \item For some $\alpha >0$: $\sup_{t\ge 1} \left|\text{cov} \left\{Z_{i,t}; Z_{j,t+s}\right\}\right| < s^{-\alpha}$. \label{var2}
\end{enumerate}
Defining $\{Y_{i,j}(t) = Z_{i,t} - Z_{j,t}, t\in \mathbb N\}$, for $i\neq j$, by (\ref{var2}) 
we have that  $\mathbb E (Y_{i,j}(t)^2)>0$.
Furthermore we assume the following condition on the temporal dependence of $Y_{i,j}(t)$:
\begin{equation}
  \label{eq:3}
  \inf_{i\neq j} \left|\sum_{t=1}^T Y_{i,j}(t)\right| \mathop{\longrightarrow}_{T\to\infty} + \infty
  \qquad \text{in probability},
\end{equation}
which means that for any $0<M<+\infty$:
\begin{equation*}
 \lim_{T\to\infty}\mathbb P\left( \inf_{i\neq j} \left|\sum_{t=1}^T Y_{i,j}(t)\right| \ge M\right) = 1.
\end{equation*}
% with probability 1 (cf. e.g. \cite{Hu and Weber 2009}).
There are many explicit conditions that imply \eqref{eq:3} for different reasons:
\begin{itemize}
\item If $\mathbb E(Z_{i,t}) \neq \mathbb E(Z_{j,t})$ for all $i\neq j$.
\item If $\{Y_{i,j}(t)\}_t$ are \emph{martingale increments}, i.e. $\mathbb E(Y_{i,j}(t)| Y_{i,j}(s), s<t) = 0$.
\item If $\{Y_{i,j}(t)\}_t$
  is strongly mixing satisfying a condition for a limit theorem (cf. \cite{Ibragimov and Linnik 1971}).
\end{itemize}

Define $X_i(T) = \sum_{t=1}^T Z_{i,t}$,  $W_i(T) = W_{i}(0)\; e^{X_i(T)}$ and
\begin{equation*}
  \theta_{i}(T) = \frac{W_{i}(T)}{\sum_j W_j(T)}, \qquad i=1, \dots, n,
\end{equation*}
and $\tilde  \theta^T_{i}(T)$ the corresponding ranked sequence at time $T$:
$\tilde\theta_{i}(T) \le \tilde\theta_{i+1}(T)$.
% $E Z = u$ and positive finite variance $0 < E Z^2 < +\infty$.
% No {\color{red} other} restriction is required on the distribution of these variables.
% Moreover, we define
% $$
% X = \sum { Z } ; W_{i,T} = \exp { X_{i,T} }, 
% \theta_i (T) = … 	for i = 1 … n
% $$
% And TETA/tilde the corresponding ranked sequence.
This indicator ranks all the relative individual wealth $W_i(T)$. % at a certain point of time t.
In particular $\tilde W^T_i(0)$ is the initial position for the sequence ranked at time $T$.
\begin{thm} Under the above conditions, the following limits holds in probability:
  \begin{equation}
    \label{eq:1}
    \lim_{T\to\infty} \tilde\theta_{n}(T) = 1, \qquad  \lim_{T\to\infty} \tilde\theta_{i}(T)  = 0 \quad \forall i\neq n  
\end{equation}
\end{thm}

\begin{proof}
Let $\tilde X_i(T), \tilde W_i(T)$ be the corresponding ranked sequences at time $T$,
  \begin{equation*}
    \begin{split}
     \tilde\theta_{n}(T) = \left(1 + \sum_{j\neq n}\tilde W_j(T) \tilde W_n(T)^{-1}\right)^{-1}
     = \left(1 + \sum_{j\neq n} \frac{\tilde W^T_j(0)}{\tilde W_n^T(0)} e^{\tilde X_j(T) - \tilde X_n(T)}\right)^{-1} \\
     \ge  \left(1 + n \left(\sup_{j\neq i} \frac{ W_j(0)}{ W_i(0)} \right)
       e^{-\inf_{j\neq n}(\tilde X_n(T) - \tilde X_j(T))}\right)^{-1}.
\end{split}
\end{equation*}
Since $W_j(0)>0$ for all $j$ we have that $\sup_{j\neq i} \frac{ W_j(0)}{ W_i(0)} < \infty$. Furthermore 
\begin{equation*}
  \inf_{j\neq n} \left(\tilde X_n(T) - \tilde X_j(T)\right) \ge \inf_{i\neq j} |X_i(T) - X_j(T)| =
  \inf_{i\neq j} |\sum_{t=1}^T Y_{i,j}(t)|,
\end{equation*}
that implies $ \mathbb P\left(\tilde\theta_{n}(T) \ge 1-\epsilon\right)
\mathop{\longrightarrow}_{T\to\infty} 1$ for any $\epsilon>0$.
\end{proof}

% \begin{rem}
%   Notice we do not need $\mu = 0$.
% \end{rem}
% \begin{rem}
%   Proof can be easily extended for $W_i(T) = W_i(0) e^{X_i(T)}$, for arbitrary initial conditions $\{W_i(0)\}_i$.
% \end{rem}

\begin{rem}
  Notice that the convergence is only in probability, and not almost sure. This means that
  a single realization of a trajectory of  $(\tilde\theta_{n}(t), t\ge 0)$ oscillates almost surely,
  but the probability to find $\tilde\theta_{n}(t)$ strictly smaller than one converges to 0 as $t\to \infty$.
  More precisely, a more accurate analysis shows that
  \begin{equation}
    \label{eq:5}
    \limsup_{T\to\infty} \tilde\theta_{n}(t) = 1, \qquad \liminf_{T\to\infty} \tilde\theta_{n}(t) = 0, \quad
    \text{almost surely.}
  \end{equation}

\end{rem}

This result implies that financial accumulation through time by individual investors leads
to ever-increasing wealth concentration and inequality across them most of the time.
This result holds for whatever initial wealth distribution and whatever mean return for financial investments over time.

The result does not imply that one single individual will own all the wealth indefinitely, but in a large time scale it can move from one investor to another. The result says that most the time wealth remains
concentrated in the hand of one individual.
%It {critically} depends on the existence of some random variance between individual returns, and their independence through time.
% {In fact much less than independence in time is sufficient. In fact there are two cases:
%   \begin{itemize}
%   \item A trivial case is when $\mu_i  \neq \mu_j$ for any $i\neq j$.
%     Then $|\sum_{t=1}^T (Z_{i,t} - Z_{j,t}) |\mathop{\longrightarrow}_{T\to\infty} + \infty$ for any $i\neq j$.
%   \item if  $\mu_i  = \mu_j$ for some $i\neq j$, then it is sufficient to assume any
%     mixing condition in time
%     that implies $|\sum_{t=1}^T (Z_{i,t} - Z_{j,t}) |\mathop{\longrightarrow}_{T\to\infty} + \infty$.
%     For example, conditions for the validity of the local central limit theorem, which can be expressed
%     in terms of the generating function $\varphi_T(k) = \mathbb E\left(e^{ik [X_i(T) - X_j(T)]}\right)$
%     and its first two derivatives in $k$ (cf. e.g. \cite{CancriniandOlla}). 
%   \end{itemize}
% }.

\section{Model and proof under simple return structure}
\label{sec:model-proof}

The previous section proved the link between financial accumulation through time and ever-increasing wealth concentration across individual investors. This section provides a complementary proof by studying simple return structure. The latter excludes financial accumulation from the financial investment process. Financial proceeds are then consumed or held for precaution (if positive), and they are not reinvested together with the previous capital stock of wealth. Simple return structure corresponds to the historical cost accounting rule of financial performance, which assumes an additive process through time between invested capital stock and generated income flow. Therefore the capital stock is reproduced more than accumulated. Financial capital is remunerated as a productive factor but not financially accumulated over time. We formalise this additive process as follows:
\begin{equation}
  W_i(T) = W_i(0) \left( 1+ \sum_{t=1}^T r_{i,t}\right) 
  \label{eq:4}
\end{equation}
assuming that $r_{i,t} >= -1$. %Since $r_{i,t} >= -1$, then $Y \ge 0$.
For sake of simplicity, we do not impose further restrictions to make this investment process more realistic.
Note that individual wealth may become negative at some point of time. 
We include only conditions that are required for mathematical proof.
This formalisation defines $\{ \mathbf r_t = (r_{1,t}, \dots, r_{n,t}),  t = 1,\dots,T\}$
as {an ergodic sequence of positive random vectors} with finite expectation:
$E r_{i,t} = \mu_i < +\infty$.
\begin{thm}
With probability one, we have
\begin{equation}
  \label{eq:2}
  \lim_{T\to\infty} \theta_{i}(T) = \ \frac{W_i(0)\mu_i}{\sum_{ j=1}^n W_j(0) \mu_j} .
\end{equation}
\end{thm}

\begin{proof}
  This is a direct consequence of the ergodic theorem. In fact with probability one
  \begin{equation*}
    \theta_{i}(T) \ =\ \frac{W_{i}(T)/T}{\sum_j W_j(T)/T} \ \mathop{\longrightarrow}_{T\to\infty}
    \ \frac{W_i(0)\mu_i}{\sum_ j W_j(0) \mu_j} .
%\ =\ \theta_i(0) . 
  \end{equation*}
\end{proof}
\begin{cor}
  If $\mu_j = \mu$ for any $j$, then $ \theta_{i}(T) \mathop{\longrightarrow}_{T\to\infty}\theta_i(0)$.
\end{cor}

\begin{rem}
  No condition is required on the correlations between the components of
  the random vector $\mathbf r_t = (r_{1,t}, \dots, r_{n,t})$.
\end{rem}

  \begin{rem} Notice that
  the convergence in \eqref{eq:2} is with probability one, stronger than the convergence in probability
  in  \eqref{eq:1}.
\end{rem}

\section{ Limits and implications}
\label{sec:limits-implications}

Our result concerning wealth concentration are twofold. In the simple interest return scenario, with probability one wealth concentration converges to certain values determined by the initial conditions. In constrast, in the compound interest return scenario, the probability that maximal wealth concentration does not occur tends to zero.
More precisely, most of the time, there exists one individual who owns all the wealth,
though not always the same individual through time: in a very large time scale
wealth can move from one lucky individual to another, passing through a fast redistribution. 

Our result implies therefore that financial accumulation through time implies ever-increasing wealth concentration and inequality across individuals.
Wealth distribution tends to degenerate so that, at a certain point of time,
a single individual eventually owns virtually all wealth, relatively speaking.
This result holds for whatever initial wealth distribution, and it critically depends
on some variance in financial returns and some return decorrelation through time.
Moreover, this result does not depend on a specific distribution of returns as long as
it involves some finite positive variance of them.
Therefore, even efficient financial markets -- implying a fair investment game
that individual investors do repeatedly play through time -- involve ever-increasing wealth inequality.
{Since this result depends only on some temporal decorrelation through time,}
it may be extended to include mutual investment opportunities that deliver
a joint return to subsets of individual investors at one point of time.
Furthermore, our analysis may be extended by studying the likely transition
time patterns that generate wealth concentration and inequality over time.

Our analysis identifies the financial accumulation process as an important driver of increasing wealth inequality in economy and society. This effect does not depend on the condition that capital investment is remunerated as a productive factor, but on the peculiar opportunity that its financial proceeds are systematically reinvested, compounding financial returns. Ever-increasing inequality depends on some temporal decorrelation between individual return trajectories.

Further mechanisms may be introduced which counter this wealth inequality effect. Countering forces may include aggregate mechanisms that redistribute wealth $W_{i,t}$ across individuals at some point of time, compensating for increasing wealth concentration and inequality over time. An obvious candidate mechanism may be taxation.
Another countering force may emerge by reaching some limit to economic process. Although generally assumed by economic models, it may be unrealistic to assume compound returns that last indefinitely over time and circumstances
(\cite{Voinov and Farley 2007,Biondi 2011}). This implicitly maintains constant returns to scale for whatever involved size, while some decreasing returns may occur after some threshold due to natural and artificial limits to growth. Furthermore, recent studies suggest that higher inequality may reduce growth, limiting financial investment opportunities
(\cite{Berg and Ostry 2013,Ostry and al. 2014}). However, according to our analysis, as long as compound returns apply, financial accumulation at those returns involves ever-increasing wealth concentration and inequality.

\section{Concluding remarks}
\label{sec:concluding-remarks}

Our approach focalises only on one dimension of economic process, namely its capital factor and related financial investment through time and circumstances. Through formal modelling, we disentangle its featuring mechanisms and study financial investment process respectively under compound and simple return structures. By construction, our modelling strategy neglects interaction with other dimensions that may affect its impact on inequality. Mathematical proofs identify the set of conditions under which financial accumulation generates ever-increasing wealth concentration and inequality. Let alone without countering forces, compound return structure constitutes then an important driver for inequality in economy and society.

Recent evolution in economy and society has been featured by increasing financial investment opportunities through active capital markets around the world. Our result corroborates that this phenomenon may have contributed to increase wealth concentration and inequality across individuals in recent decades.

From an epistemological viewpoint, our approach points to the interaction between microscopic and macroscopic dimensions of the economic system under investigation. One cumulative effect concerns compound financial returns through time at the level of every individual investment portfolio. Another cumulative effect concerns wealth concentration across population of individual investors at every point of time.

Our approach contributes therefore to include space heterogeneity and time evolution into economic modelling. Concerning time, financial accumulation has been somewhat neglected because modelling strategies were considering one-period or two-period horizons, while the compound interest structure plays its featuring role when at least three successive periods are considered. Concerning space, our approach introduces a population of heterogeneous individuals. Individual financial investment trajectories constitute a heterogeneous set of financial accumulation processes performed under active financial markets through time and circumstances.


\begin{thebibliography}{99}

\bibitem[Beck et al. 2007]{Beck et al. 2007} Beck, T., Demirguc-Kunt, A. and Levine, R. (2007),
  Finance, inequality and the poor, Journal of Economic Growth 12(1), 27-49.
  
\bibitem[Berg and Ostry 2013]{Berg and Ostry 2013} Berg, A. G., Osrty, J. D. (2013).
  Inequality and unsustainable growth: Two sides of the same coin?,
  International Organizations Research Journal, 8 (4):77-99.

\bibitem[Bertola et al. 2006]{Bertola et al. 2006}
  Bertola,  G.,  Foellmi,  R.,  Zweimueller,  J.  (2006)  Income  distribution  in  macroeconomic models.
  Princeton (NJ): Princeton University Press.
  
\bibitem[Biondi 2011]{Biondi 2011} Biondi, Y. (2011).
  Cost of capital, discounting and relational contracting: endogenous optimal return and duration
  for joint investment projects. Applied Economics, 43 (30):4847-4864.
  
\bibitem[Biondi and Righi 2019]{Biondi and Righi 2019} Biondi, Y., Righi, S. (2019),
  Inequality, Mobility and the Financial Accumulation Process:
  A Computational Economic Analysis.
  J Econ Interact Coord (2019) 14: 93. DOI: https://doi.org/10.1007/s11403-019-00236-7
  % 2nd International Workshop on “Financial Markets and Nonlinear Dynamics” (FMND),
  % Paris, 4-5 June 2015.
 % DOI: http://dx.doi.org/10.2139/ssrn.2628536
  
\bibitem[Biondi and Righi 2017]{Biondi and Righi 2017} Biondi, Y., Righi, S. (2017)
  Much ado about making money: the impact of disclosure,
  news and rumors on the formation of security market prices over time,
  Journal of Economic Interaction and Coordination. DOI: 10.1007/s11403-017-0201-8.

% \bibitem[Cancrini and Olla 2017]{CancriniandOlla} Cancrini, N. and Olla, S.,
%   Ensemble Dependence of Fluctuations,
%   Journal of Statistical Physics, 2017 vol. 168 (4) pp. 707-730,
%   doi:10.1007/s10955-017-1830-y.
  
\bibitem[Champernowne 1953]{Champernowne 1953} Champernowne, D. G. (1953)
  A model of income distribution. The Economic Journal, Vol. 63, No. 250, June: 318–351.
  
\bibitem[Fama 1995]{Fama 1995} Fama, E.F. (1995)
  Random walks in stock market prices. Financial Analysts Journal, 21 (5) September - October: 55-59.
  
\bibitem[Gini 1912]{Gini 1912} Gini, C. (1912) Variabilit\`a e mutabilit\`a.
  Reprinted in: Memorie di metodologica statistica (Edited by E. Pizetti and T. Salvemini) Rome,
  Libreria Eredi Virgilio Veschi.

% \bibitem[Hu and Weber 2009]{Hu and Weber 2009} Hu, T.C., Weber, N.C., (2009)
%   A Note on Strong Convergence of Sums of Dependent Random Variables.
%   J. of Probability and Statistics, Article ID 873274, 7 pages, 2009. https://doi.org/10.1155/2009/873274.

\bibitem[Ibragimov and Linnik 1971]{Ibragimov and Linnik 1971} Ibragimov and Linnik, (1971)
  Independent and stationary sequences of random variables,
  Walters-Noordhoff, Groningen.
  
\bibitem[Kalecki 1945]{Kalecki 1945} Kalecki, M. (1945)
  On the Gibrat Distribution, Econometrica, 13 (2): 161-70
  
\bibitem[Kesten 1973]{Kesten 1973} Kesten, H. (1973)
  Random difference equations and Renewal theory for products of random matrices,
  Acta Mathematica, Volume 131, 207-248.
  
\bibitem[Krugman  2013]{Krugman  2013} Krugman, P. (2013).
  ‘Why inequality matters’. International New York Times, 15 December.
  
\bibitem[Levy 2005]{Levy 2005} Levy, M. (2005) ‘Market efficiency, the Pareto wealth distribution,
  and the Levy distribution of stock returns’.
  In: The Economy As an Evolving Complex System, III: Current Perspectives and Future Directions.
  Oxford (UK): Oxford University Press.
  
\bibitem[Levy and Levy 2003]{Levy and Levy 2003} Levy, M., Levy, H. (2003)
  Investment talent and the Pareto wealth distribution:
  Theoretical and experimental analysis.
  Review of Economics and Statistics 85(3):709–725.
  
\bibitem[Nirei and Souma 2007]{Nirei and Souma 2007} Nirei, M., Souma, W. (2007)
  A Two Factor Model of Income Distribution Dynamics, Review of Income and Wealth, 53 (3) September: 440-459.
  
\bibitem[Ostry and al. 2014]{Ostry and al. 2014} Ostry, M. J. D., Berg, M. A., Tsangarides, M. C. G. (2014).
  ‘Redistribution, Inequality, and Growth’. International Monetary Fund.
  
\bibitem[Pareto 1897]{Pareto 1897} Pareto, V. (1897). Cours d'Économie Politique. Lausanne: F. Rouge.
  
\bibitem[Piketty 2013]{Piketty 2013} Piketty, Th. (2013). Le Capital au XXIe siècle, Paris: Le Seuil.
  
\bibitem[Redner 1990]{Redner 1990} Redner, S. (1990) Random multiplicative processes: An elementary tutorial,
  Am. J. Physics, 58 (3) March: 267-273.
  
\bibitem[Rutherford 1955]{Rutherford 1955} Rutherford, R. S. G. (1955)
  Income Distribution: A New Model, Econometrica, 23: 277-94.
  
\bibitem[Samuelson 1965]{Samuelson 1965} Samuelson P.A. (1965) Proof that properly anticipated prices fluctuate randomly. Industrial Management Review, 6(2):41-49.
  
\bibitem[Samuelson 1973]{Samuelson 1973} Samuelson, P.A. (1973) Proof that properly discounted present values of assets vibrate randomly. Bell Journal of Economics and Management Science, 4 (2): 369–374.
  
\bibitem[Snowdon and Vane 2005]{Snowdon and Vane 2005} Snowdon, B., Vane, H.R. (2005) Modern macroeconomics: its origins, development and current state. Cheltenham (UK): Edward Elgar Publishing.
  
\bibitem[Solow 2014]{Solow 2014} Solow, R.M. (2014)
  The rich-get-richer dynamic the actual economics of inequality. New Republic, 245 (8):50–55.
  
\bibitem[Stiglitz 2012]{Stiglitz 2012} Stiglitz, J. (2012) The price of inequality.
  London: Penguin.

\bibitem[Voinov and Farley 2007]{Voinov and Farley 2007} Voinov, A., Farley, J. (2007).
  Reconciling sustainability, systems theory and discounting. Ecological Economics, 63 (1):104-113.
  
\end{thebibliography}
  \end{document}